\documentclass{article}
\usepackage[utf8]{inputenc}
\usepackage{authblk}
\usepackage{physics}
\usepackage{todonotes}
\usepackage{color}

\usepackage{amsthm,amssymb,mathtools,amsmath}

\newtheorem{theorem}{Theorem}

\newcommand{\comment}[1]{}

\begin{document}

%------------------------------------------------------------------------------------

\title{Comment on Nahimovs et al. `On the probability of finding marked connected components using quantum walks'}
\author[1,2]{Adam Glos}
\author[3]{Nikolay Nahimovs}
\affil[1]{Institute of Theoretical and Applied Informatics, Polish Academy of 
		Sciences, ul. Ba{\l}tycka 5, 44-100 Gliwice, Poland}
	\affil[2]{Institute of Informatics, Silesian University of Technology,
		ul. Akademicka 16, 44-100 Gliwice, Poland}
\affil[3]{Centre for Quantum Computer Science, Faculty of Computing University of Latvia Riga Latvia}
\date{}

\maketitle

%------------------------------------------------------------------------------------

\begin{abstract}
	In this comment paper we present two misconceptions found in paper of Nahimovs et al. \emph{On the probability of finding marked connected components using quantum walks}. First, we show that the Theorem 2 (sufficient and necessary condition for a state to be stationary) is incomplete -- it works only if unmarked vertices form a single connected component. Second, we correct derivation of \emph{a} coefficient in the Theorem 3 (lower bound on the probability) and show how to upper bound value of \emph{a}.
\end{abstract}

%------------------------------------------------------------------------------------

\section{Corrections}

\subsection{Theorem 2}

Let $G=(V,E)$ be an undirected graph and $M$ be a connected set of marked vertices. 
In \cite{prusis2016stationary} the authors showed
\begin{theorem}
\label{thm:shortages}
	If $M$ is bipartite,  then we  can  assign amplitudes to neutralise the shortages at each marked vertex if and only if the sums of the shortages on both partite sets are equal.
\end{theorem}
The shortage of a marked vertex is defined to be a sum of amplitudes between the vertex and its
neighbouring unmarked vertices. Neutralising the shortages means finding amplitudes inside the marked component such that the total sum of amplitudes of each marked vertex is $0$. 

Let $M$ be bipartite with partite sets $M_1$ and $M_2$ and let
\begin{equation}
\deg_G(M_i) \coloneqq \sum_{v\in M_i} \deg_G(v),
\end{equation}
be a total outgoing degree of $M_i$, where $\deg_G(v)$ is the degree of $v$ in $G$.
The authors of \cite{khadiev2018probability} claimed that according to the theorem above, the following theorem holds.
\begin{theorem}
	A bipartite marked connected component has a stationary state if and only if $\deg_G(M_1) = \deg_G(M_2)$. A non-bipartite marked connected component always has a stationary state.
 \end{theorem}
It turns out that the theorem is correct only if the induced subgraph spanned by
$V \setminus M$ is connected, as in such scenario all amplitudes of edges incident to an unmarked vertex must be equal. 
However, if $V \setminus M$ is disconnected then amplitudes in its connected components might be different. Therefore, one can always construct a stationary state.

\begin{theorem}
	Let $G=(V,E)$ be connected graph, let $M$ be a connected set of marked vertices and let $V \setminus M$ be disconnected. Then $G$ always has a stationary state.
\end{theorem}

\begin{proof}

For simplicity suppose $V \setminus M$ consist of two connected components $H_1$ and $H_2$. The argument below can be easily extended to arbitrary number of components.

Let $k_{11} = | \{ (i,j) \in E \: | \: i \in H_1, j \in M_1 \} |$ be a number of edges between vertices of $H_1$ and $M_1$ (see Figure \ref{fig:two_components}). Similarly, we define $k_{12}$, $k_{21}$ and $k_{22}$. Let amplitudes of edges incident to vertices of $H_1$ be $a_1$ (all amplitudes must be equal) and amplitudes of edges incident to vertices of $H_2$ be $a_2$.
\begin{figure}
    \centering
    \includegraphics[scale=0.5]{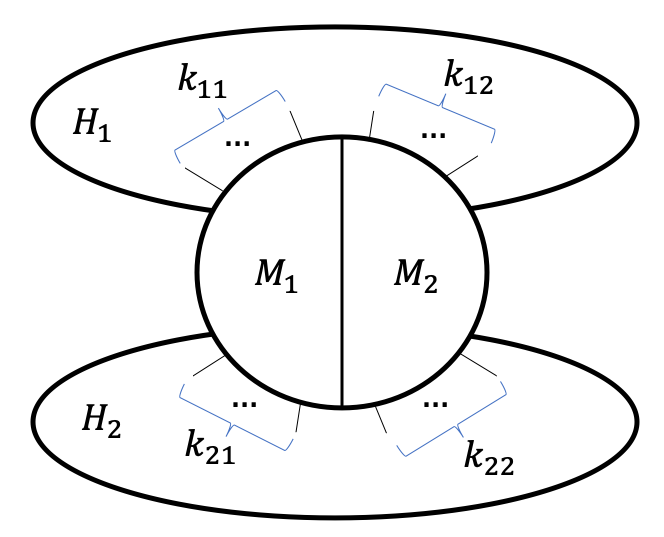}
    \caption{$V \setminus M$ consisting of two connected components $H_1$ and $H_2$}
    \label{fig:two_components}
\end{figure}
The sum of shortages of $M_1$ is
$$
s(M_1) \coloneqq k_{11} \cdot a_1 + k_{21} \cdot a_2
$$
and the sum of shortages of $M_2$ is
$$
s(M_2) \coloneqq k_{12} \cdot a_1 + k_{22} \cdot a_2 .
$$
It follows from Theorem \ref{thm:shortages} that if $s(M_1) = s(M_2)$ one can assign amplitudes to neutralise shortages (i.e. construct a stationary state). Therefore, we need
$$
k_{11} \cdot a_1 + k_{21} \cdot a_2 = k_{12} \cdot a_1 + k_{22} \cdot a_2
$$
or
\begin{equation}
\label{eq:a1_a2}
(k_{11} - k_{12}) \cdot a_1 = (k_{22} -  k_{21}) \cdot a_2 .
\end{equation}
As $a_1$ and $a_2$ are independent it is always possible to choose their values to satisfy the equality.
\end{proof}

Note that the constructed stationary state can have $0$ overlap with the initial state.
As before, let  $V \setminus M$ consist of two connected components $H_1 =(V_1, E_1)$ and $H_2=(V_2, E_2)$.

The initial state has all amplitudes equal to $\frac{1}{\sqrt{2m}}$. The overlap of $M$ with the initial state is $0$ (as the sum of amplitudes for each vertex of $M_1$ is $0$).
Consider overlap of the rest of the graph $|E_1| \cdot a_1 + |E_2| \cdot a_2$.
Let $k_{11} = k_{21}$ and $k_{12} = k_{22}$. 
Then, from \eqref{eq:a1_a2} we have $a_1 = -a_2$.
Therefore, if $|E_1| = |E_2|$ we have $|E_1| \cdot a_1 + |E_2| \cdot a_2 = 0$.

%------------------------------------------------------------------------------------

\subsection{Derivation of $a$ in the proof of Theorem 3.}

In the proof of Theorem 3 in \cite{khadiev2018probability} authors claimed that $a = \frac{1}{\sqrt{2m}}$, where $m$ is the number of edges. Authors have defined $a$ through stationary state $\ket{\psi_{ST}^a}$ and have not considered the normalisation of the state, which in turns makes the mentioned equality incorrect in general. While it may be difficult to precisely determine the optimal $a$ value, it is relatively easy to provide upper bound that is satisfactory for small exceptional configurations.

Note that authors have considered only stationary states with the amplitudes $u \to v$ for arbitrary unmarked $u$  and arbitrary $v$ being all equal. Since according to Theorem 2 from \cite{prusis2016stationary} only these amplitudes has impact on overlap between stationary state and initial state, $a$ can be upper-bounded by $\bar a$ of the form
\begin{equation}
\ket{\psi} = \sum_{u \in V \setminus M} \sum_{v \in N(u)} \bar a \ket{u,v}.
\end{equation}
Note that $\ket{\psi}$ does not need to be a stationary state, but presents the worst case scenario of optimal stationary states. Thus
\begin{equation}
\begin{split}
1  = \braket{\psi} &= \sum_{u\in V\setminus M} \sum_{v \in N(u)} \bar a^2\\
&= \bar a^2\sum_{u\in V\setminus M} \deg_G(v)\\
&= \bar a^2(2m - 2|E_M|-D^{\bar M}),
\end{split}
\end{equation}
where $E_M$ is the set of edges between marked vertices ad $D^{\bar M}$ is the number of edges between marked and unmarked vertices. Thus,
\begin{equation}
a \leq \bar a  = \frac{1}{\sqrt{2m - 2|E_M|-D^{\bar M}}}.
\end{equation}

\section{Final statements of Theorems 3}
Finally we present a corrected version of Theorem 3.
\begin{theorem}
Consider a graph $G=(V,E)$ with a component of marked vertices $M$. Let $M$ be such that there exists a stationary state with amplitudes of unmarked$\to$marked arcs being all equal. Then the probability $p_M(t)$ of finding a marked vertex after $t$ steps satisfies 
\begin{equation}
p_M \leq \frac{4}{2m-2|E_M|-D^{\bar M}} \Big ( \sum_{\substack{i,j\in M \\i\sim j}}c_{ij}^2 + 2D^{\bar M} + 2|E_M|\Big ).
\end{equation}
\end{theorem}
Note that the Theorem 3 covers also the Corollary 1 presented in \cite{khadiev2018probability}.

%------------------------------------------------------------------------------------

\bibliographystyle{ieeetr}
\bibliography{comment}

\begin{thebibliography}{1}

\bibitem{prusis2016stationary}
K.~Pr{\=u}sis, J.~Vihrovs, and T.~G. Wong, ``Stationary states in quantum walk
  search,'' {\em Physical Review A}, vol.~94, no.~3, p.~032334, 2016.

\bibitem{khadiev2018probability}
K.~Khadiev, N.~Nahimovs, and R.~Santos, ``On the probability of finding marked
  connected components using quantum walks,'' {\em Lobachevskii Journal of
  Mathematics}, vol.~39, no.~7, pp.~1016--1023, 2018.

\end{thebibliography}

\end{document}